\documentclass[10pt]{article}
\pdfoutput=1
\usepackage[lmargin=1.2in,rmargin=1.2in,tmargin=1.2in,bmargin=1.8in]{geometry}
\usepackage{amsthm,amsfonts,amssymb,amsmath,comment,slashed,mathtools,bbold}
\usepackage[T1]{fontenc} %
\usepackage[utf8]{inputenc}
\usepackage{bbm}
\usepackage[affil-it]{authblk} %
\usepackage{hyperref}    
\hypersetup{hidelinks}
\usepackage{graphicx} 
\usepackage{subcaption}
\usepackage{thm-restate} %
\usepackage{xargs}          
\usepackage{here}
\usepackage{mathrsfs}
\usepackage{xcolor}  %
\usepackage{here} %
\usepackage[capitalize]{cleveref} %
\usepackage{tikz-cd} %
\usepackage{color}
\usepackage{transparent}
\usepackage{bm}
\usepackage[backend=biber,bibencoding=ascii,giveninits=true,doi=false,isbn=false,url=false,sorting=nyt,maxbibnames=99]{biblatex} 
\renewbibmacro{in:}{} %
\bibliography{references.bib}
\theoremstyle{plain} 
 
\newtheorem{prop}{Proposition}

\newtheorem{lemma}{Lemma}

\newtheorem{rmk}{Remark}

\newtheorem{conjecture}{Conjecture}

\declaretheorem[name=Theorem]{theorem}
\renewcommand{\d}{\mathrm{d}}

\newcommand{\Ch}{\mathcal{CH}}

\newcommand{\dvol}{\mathrm{dvol}}

\newcommand{\uhr}{{u_{\mathcal{H}_R}}}
\newcommand{\uchl}{{u_{\mathcal{CH}_L}}}
\newcommand{\uchr}{{u_{\mathcal{CH}_R}}}

\makeatletter
\renewcommand{\paragraph}[1]{%
	\par %
	\addvspace{\medskipamount}%
	\textbf{\textit{#1\@addpunct{.}}}\enspace\ignorespaces
}
\makeatother
\numberwithin{equation}{section}
\numberwithin{prop}{section}
\numberwithin{rmk}{section}
\numberwithin{lemma}{section}
\numberwithin{definition}{section}

\title{Blowup of the local energy of linear waves  at~the~Reissner--Nordström--AdS Cauchy~horizon}
\author{Christoph Kehle\thanks{c.kehle@eth-its.ethz.ch}}
\affil{\small  
	Institute for Theoretical Studies, ETH Zürich, Clausiusstrasse 47, 8092 Zürich, Switzerland \vskip.1pc \  }

\date{August 9, 2021}
\begin{document}
	
\maketitle
\thispagestyle{empty}
\begin{abstract}
We show that generic linear perturbations $\psi$ solving $\Box_g \psi + \frac{\alpha}{l^2} \psi =0$ on Reissner--Nordström--AdS black holes have infinite local energy at the Cauchy horizon. Combined with the  result of \cite{CKehle2019} that such perturbations $\psi$ remain uniformly bounded and extend continuously across the Cauchy horizon, this settles the linear analog of the Strong Cosmic Censorship conjecture for Reissner--Nordström--AdS: the $C^0$-formulation is false but the $H^1$-formulation is true. 
\end{abstract}

\section{Introduction}
The purpose of this paper is to show blowup at the Reissner--Nordström--AdS  Cauchy horizon of the \emph{local energy} of linear perturbations $\psi$  solving the Klein--Gordon equation
\begin{align}\label{eq:waveequation}
	\Box_{g} \psi + \frac{\alpha}{l^2} \psi =0.
\end{align}
We   impose that $\psi$ arises from generic regular Cauchy data and satisfies reflecting boundary conditions at infinity $\mathcal I$. 
Here, $l^2:=\frac{3}{-\Lambda}$ is the AdS radius for a  fixed negative cosmological constant $\Lambda <0$ and $\alpha\in \mathbb R$ satisfies the Breitenlohner--Freedman  bound $\alpha < \frac 94$ \cite{breitenlohner}. The problem is motivated by the \emph{Strong Cosmic Censorship (SCC) conjecture} due to Penrose \cite{penrose1974gravitational}. Our main result proving blowup of the local energy of $\psi$ at the Cauchy horizon shows  the $H^1$-formulations of the linear analog of the Strong  Cosmic Censorship conjecture for Reissner--Nordström--AdS  (see already \cref{thm:maintheomintro}). 

In \cref{subsec:intro1} we begin by recalling the Strong Cosmic Censorship conjecture. We further present the $C^0$- and $H^1$-formulation as well as outline the state of the art for the   $\Lambda \geq 0$ cases. In \cref{subsec:thestateoftheartforlambda<0}  we   focus on the  $\Lambda <0$ case and introduce the linear analogs of the $C^0$- and $H^1$-formulation of SCC. We then discuss the status of the $C^0$-formulation before we state our main result \cref{thm:maintheomintro}.

\subsection{Smooth Cauchy horizons and Strong Cosmic Censorship}
\label{subsec:intro1}
The Reissner--Nordström--AdS black holes constitute a family of solutions to the Einstein--Maxwell system 
\begin{align}\label{eq:einstein}
R_{\mu \nu} [g] = \Lambda g_{\mu\nu} + 2 \left( F_{\mu}^\lambda F_{\lambda \nu} - \frac 14 g_{\mu \nu} F_{\lambda \kappa} F^{\lambda \kappa} \right), \;\; \; \nabla^\nu F_{\mu \nu} = 0, \;\;\; \nabla_{[\mu} F_{\nu \lambda]}
\end{align}
in the presence of a negative cosmological constant $\Lambda <0$. The Reissner--Nordström--AdS Cauchy horizon---representing the future boundary of the interior of a Reissner--Nordström--AdS black hole---has the puzzling feature that spacetime is  smoothly but non-uniquely extendible across it. This is puzzling as physical observers entering the interior of the black hole  reach and cross the Cauchy horizon in finite proper time beyond which their fate is  not determined.  In that sense, smoothness of  Cauchy horizons may be considered ``good'' for the observers themselves (as they are not destroyed!) but ``bad'' for global predictability of the theory of general relativity. 
\paragraph{The SCC conjecture}
 A way out of the puzzle of the existence of smooth Cauchy horizons was proposed by Penrose who identified a blue-shift instability \cite{Penrose:1968ar} associated to them, suggesting that smooth Cauchy horizons may be destroyed upon perturbation.  This led Penrose to the celebrated  Strong Cosmic Censorship conjecture \cite{penrose1974gravitational}  according to which  \emph{generic} asymptotically flat initial data sets for the Einstein equations \eqref{eq:einstein} give rise to   maximal spacetimes which are \emph{inextendible} as   \emph{suitably regular} Lorentizan manifolds. While the conjecture was originally only stated for asymptotically flat $\Lambda =0$ solutions to \eqref{eq:einstein}, a similar conjecture can be stated for $\Lambda \neq 0$.
 \paragraph{The most definitive formulation: the $C^0$-formulation of SCC}  Part of addressing the SCC conjecture is to specify what is meant by ``suitably regular''. The most definitive answer to the issue of Strong Cosmic Censorship would be if generically the metric $g$ can be shown to be inextendible as a continuous metric---the so-called $C^0$-formulation of Strong Cosmic Censorship. This corresponds to the scenario where physical observers reaching the boundary of spacetime are destroyed by infinite tidal deformations.

 \paragraph{Downfall of the $C^0$-formulation for $\Lambda \geq 0$} For the cases $\Lambda =0$ and $\Lambda >0$, Dafermos--Luk proved in \cite{dafermos2017interior} that perturbations of Kerr or Kerr--de~Sitter admit a Cauchy horizon beyond which the metric is continuously extendible. This falsifies the $C^0$-formulation  for the vacuum Einstein equations. Their work \cite{dafermos2017interior} was preceded by the falsification of a linear analog of the $C^0$-formulation of SCC for \eqref{eq:waveequation} in   \cite{anneboundedness,annekerr,hintzinterior} (see already \cref{conj:linearanalogofC0} for the linear analog of the $C^0$-formulation in the case $\Lambda <0$). Refer also to  \cite{dafermos_cauchy_horizon} in   spherical symmetry and the pioneering \cite{mcnamara1978behaviour}. 
 (In the presence of matter, the situation is less clear and the $C^0$-formulation may even be true, see \cite{kehle2021strong} for \eqref{eq:einstein} coupled to  a charged/massive scalar field in spherical symmetry.) 

\paragraph{The weaker formulation: the $H^1$-formulation of SCC} 
 While the $C^0$-formulation has been shown to be false for $\Lambda \geq 0$ in vacuum, the weaker $H^1$-formulation of SCC due to Christodoulou \cite{christo}  is expected to be true (at least for $\Lambda =0$): In rough terms, the $H^1$-formulation states that generically, the metric is inextendible as a continuous metric with square integrable Christoffel symbols. This can be interpreted as blowup of the local energy.  If true, the $H^1$-formulation of SCC can be considered to at least restore a weaker version of determinism.
   For further details on the state of the art of the $H^1$-formulation or its linear analog for \eqref{eq:waveequation} (see already \cref{conj:linearanalogofh1} for the linear analog of the $H^1$-formulation in the case $\Lambda <0$), we refer to \cite{daf05,dafermos2017time,luk2016kerr,luk2017strong,luk2017strong2,lukohchblowup,vandemoortel2,VandeMoortel2018}  and the earlier works \cite{MR617171,internal90,ori91} for the case  $\Lambda =0$ and to \cite{dafermos2018rough,hintzvasyinterior,MR3697197,cardosoetal,dias2018strong,dias2018strong2,MR3882684,reallsurvey} for the case $\Lambda >0$. Refer also to the discussion at the end of \cref{subsubsec:mainres}.

\subsection{The state of the art  of SCC for \texorpdfstring{$\Lambda <0$}{Lambda <0} }
\label{subsec:thestateoftheartforlambda<0}
In the case $\Lambda <0$, i.e.\ asymptotically Anti-de~Sitter black holes, the situation turns out to be radically different and has been significantly less explored. Before we address the issue of Strong Cosmic Censorship for $\Lambda <0$, we give some comments on the analysis of waves on asymptotically AdS spacetimes.

First, in view of the timelike nature of infinity $\mathcal I$ one has to additionally impose boundary conditions at   $\mathcal I$. We will consider  Dirichlet boundary conditions    but  without much change one could also  consider Neumann or suitable reflecting Robin boundary condition. 

Secondly, with such   boundary conditions, the exteriors of AdS black holes are significantly less stable compared to the $\Lambda = 0$ and $\Lambda >0$ cases. In particular, for linear perturbations $\psi$ solving \eqref{eq:waveequation},  it was  shown that Reissner--Nordström--AdS  and Kerr--AdS (for parameters satisfying the Hawking--Reall bound \cite{dold}) are weakly stable in the exterior as such linear perturbations $\psi$ only decay at a slow inverse logarithmic rate   as shown in \cite{gustav,quasimodes} (cf.\  polynomial and exponential decay of linear perturbations for $\Lambda=0$ and $\Lambda >0$).   The slow decay is associated to a stable trapping phenomenon of high-frequency null geodesics with high angular momentum which bounce back and forth between the angular momentum barrier and     $\mathcal I$. A  manifestation of this stable trapping of null geodesics is that there exist  quasinormal frequencies $\omega_{n,m,\ell} $ and corresponding quasinormal modes (QNMs) $\psi_{n,m,\ell}$ with imaginary part converging exponentially fast to the real axis. More precisely, such  QNMs are regular solutions to \eqref{eq:waveequation} of the separated form 
\begin{align}\label{eq:qnms}
\psi_{n,m,\ell}(t,r,\theta,\varphi) = e^{-i\omega_{n,m,\ell} t} R_{n,m,\ell}(r) Y_{\ell,m}(\theta) e^{i m \varphi}
\end{align} 
for $\ell\in \mathbb N$ large, $m\in \mathbb Z$, $|m|\leq \ell$, $1 \leq n \leq N(m,\ell)$, $\omega_{n,m,\ell} \in \mathbb C$, satisfying
\begin{align}\label{eq:qnmsbounds}
	\operatorname{Re}(\omega_{n,m,\ell} ) \sim \ell, \;\;\; \operatorname{Im}(\omega_{n,m,\ell}) \sim - e^{- c_I \ell}
\end{align} for Schwarzschild--AdS \cite{MR3223487}  and more generally for Reissner--Nordström--AdS. Similarly, upon  replacing the spherical harmonics in \eqref{eq:qnms} with so-called modified spheroidal harmonics, there exist analogous QNMs satisfying \eqref{eq:qnmsbounds} on   Kerr--AdS as shown in \cite{quasinormal_gannot,quasimodes}, see also \cite{warnick_quasi,petersen2021analyticity}. 
In fact, a generic infinite superposition of quasinormal modes with finite energy merely decays at a slow inverse logarithmic rate; see also the related concept of quasimodes \cite{quasimodes}.
In any case, whether this weak linear stability can be upgraded to a full nonlinear stability statement or whether already the exteriors of AdS black holes are unstable still remains widely open, see however \cite{moschidis2018proof,moschidis2017proof,moschidis2017einstein} for pure AdS.

\subsubsection{The $C^0$-formulation for $\Lambda <0$}
Turning back to the problem of Strong Cosmic Censorship and its linear analog for \eqref{eq:waveequation}, the slow decay  of linear perturbations $\psi$ on the black hole exterior can however be seen as good news as this could mean that the most definitive $C^0$-formulation may be true after all for $\Lambda <0$. One may   thus conjecture that the following linear analog of the $C^0$-formulation of Strong Cosmic Censorship holds true for subextremal asymptotically AdS black holes. 
\begin{conjecture}[Linear analog of $C^0$-formulation of SCC]
	\label{conj:linearanalogofC0}
	Linear perturbations $\psi$ solving \eqref{eq:waveequation} on Reissner--Nordström--AdS or  Kerr--AdS and arising from generic initial data with Dirichlet (or more general) boundary conditions imposed at infinity $\mathcal I$, blow up in amplitude at the Cauchy horizon $\Ch$: $\lim_{x\to \Ch} |\psi(x)|\to  \infty$.
\end{conjecture}

\paragraph{\cref{conj:linearanalogofC0} is false for Reissner--Nordström--AdS}
In the work \cite{CKehle2019} it was shown that despite the slow decay on the exterior,   all linear perturbations $\psi$ solving \eqref{eq:waveequation} of Reissner--Nordström--AdS arising from sufficiently regular initial data and with Dirichlet boundary conditions imposed at infinity $\mathcal I$ remain uniformly bounded on the black hole interior and extend continuously across the Cauchy horizon $\Ch$. This shows
\begin{theorem}[\cite{CKehle2019}]
\cref{conj:linearanalogofC0} is false for Reissner--Nordström--AdS. 
\end{theorem}
The crucial observation of the proof is that two sources of instability which may lead to blowup in amplitude  decouple in frequency space. More precisely, we first recall that only the high frequency part ($|\omega|, \ell$ large, $|m|\leq\ell$) of the perturbation $\psi$ decays slowly on the exterior, whereas the low frequency part is in fact shown to decay superpolynomially \cite{CKehle2019}. On the other hand, in the interior, it is only the  low-frequency (measured with respect to the null generator of the Cauchy horizon) part of the perturbation $\psi$ which may lead to blowup in amplitude at the Cauchy horizon. Remarkably, the latter instability of the low frequency part of $\psi$ in the interior is not a result of the celebrated blue-shift effect but   a consequence of a pole of the scattering coefficient $\mathfrak R \sim \frac{1}{\omega}$ from the event to the Cauchy horizon as observed in \cite{kehle2018scattering}.
\paragraph{ \cref{conj:linearanalogofC0} for Kerr--AdS and its connection to Diophantine approximation}
The situation for Kerr--AdS has yet another level complexity compared to the Reissner--Nordström--AdS case. Indeed, due to the rotation of the Kerr--AdS black hole, a frequency mixing phenomenon occurs and slowly decaying high frequency solutions $\psi$ on the exterior ($|\omega|, \ell$ large, $|m|\leq \ell$) can at the same time be low frequency ($\omega - \omega_- m$ small) with respect to the generator $K_- = T + \omega_- \Phi$  of the Cauchy horizon. 
In particular, the slowly decaying high frequency  quasinormal modes with $(\omega_{\ell,m,n}, \ell, m)$  can satisfy $\omega_{\ell,m,n} - \omega_- m \approx 0$. This yields a small divisors problem as we recall that the scattering coefficient is of the form $\mathfrak R \sim \frac{1}{\omega - \omega_- m } = \frac{1}{\omega_{\ell,m,n} - \omega_- m}$. Reminiscent of small divisors problem are certain Diophantine condition which  turn out (see  \cite{kehle2020diophantine}) to play a crucial role in addressing \cref{conj:linearanalogofC0} for Kerr--AdS. To state the following theorem we denote with $\mathcal P$ the parameter space of dimensionless masses and angular momenta $(\mathfrak m, \mathfrak a) := (M \sqrt{-\Lambda}, a \sqrt{-\Lambda}  ) $ corresponding to all  subextremal Kerr--AdS black holes satisfying the Hawking--Reall bound. 
 \begin{theorem}[\cite{kehle2020diophantine}]\label{thm:baireblowup}
\cref{conj:linearanalogofC0} is true for Baire-generic Kerr--AdS black holes.  

More precisely,  linear perturbations $\psi$ solving \eqref{eq:waveequation} with $\alpha = 2$ on Kerr--AdS   blow up in amplitude at the Cauchy horizon $\lim_{x\to \Ch}|\psi(x)|\to \infty$ for dimensionless Kerr--AdS black hole parameters $(\mathfrak m, \mathfrak a) := (M \sqrt{-\Lambda}, a \sqrt{-\Lambda}  )  \in \mathcal P_{\textup{Blowup}} \subset \mathcal P$, where  
\begin{itemize}
\item $ \mathcal P_{\textup{Blowup}} \subset \mathcal P$ is dense,
\item $ \mathcal P_{\textup{Blowup}} \subset \mathcal P$ is Baire-generic (contains a countable intersection of open and dense sets),
\item $ \mathcal P_{\textup{Blowup}}\subset \mathcal P$ is a Lebesgue null set.
\end{itemize}
 \end{theorem}

On the other hand the proof of \cref{thm:baireblowup} in \cite{kehle2020diophantine} and the heuristics outlined in \cite{kehle2020diophantine}   strongly suggest   that  for   the complement of $\mathcal P_{\textup{Blowup}}$, a set which is Lebesgue-generic (but Baire-exceptional), the amplitude of $\psi$ remains uniformly bounded at the Cauchy horizon.\begin{conjecture}[\cite{kehle2020diophantine}]
	\cref{conj:linearanalogofC0} is false for Lebesgue-generic Kerr--AdS black holes.  
	
	More precisely,  linear perturbations $\psi$ solving \eqref{eq:waveequation} with $\alpha = 2$ on Kerr--AdS   remain uniformly bounded $ |\psi(x)|\leq C $ and extend continuously across the Cauchy horizon for dimensionless Kerr--AdS black hole parameters $(\mathfrak m, \mathfrak a) := (M \sqrt{-\Lambda}, a \sqrt{-\Lambda}  )  \in \mathcal P_{\textup{Bounded}} \subset \mathcal P$, where  
	\begin{itemize}
		\item $ \mathcal P_{\textup{Bounded}} \subset \mathcal P$ is dense,
		\item $ \mathcal P_{\textup{Bounded}} \subset \mathcal P$ is Lebesgue-generic (full Lebesgue measure),
		\item $ \mathcal P_{\textup{Bounded}} \subset \mathcal P$ is Baire-exceptional.
	\end{itemize}
\end{conjecture} 
Thus, for Kerr--AdS the linear analog of the $C^0$-formulation of Strong Cosmic Censorship is true if Baire-genericity is imposed and is conjectured to be false if Lebesgue-genericity is imposed.

\subsubsection{Main result:  Linear analog of the $H^1$-formulation of SCC   is true for Reissner--Nordström--AdS}
\label{subsubsec:mainres}
While the above shows that   for Reissner--Nordström--AdS (and conjecturally for Lebesgue-generic Kerr--AdS black holes), the linear analog of the $C^0$-formulation of Strong Cosmic Censorship is false, it still remains open whether at least the linear analog of the weaker  $H^1$-formulation of Strong Cosmic Censorship is true.

\begin{conjecture}[Linear analog of $H^1$-formulation of SCC]
	\label{conj:linearanalogofh1}
	Linear perturbations $\psi$ solving \eqref{eq:waveequation} on Reissner--Nordström--AdS or more generally Kerr--AdS and arising from generic Cauchy data with Dirichlet  boundary conditions imposed at infinity $\mathcal I$  have infinite local energy along hypersurfaces intersecting transversally the Cauchy horizon $\Ch$, i.e.\ $\psi$ fails to be in $H^1_{\mathrm{loc}}$ around any point on $\Ch$. 
\end{conjecture}

As our main result in this paper, we answer this question in the affirmative for Reissner--Nordström--AdS. The following theorem is a direct consequence of the more precise \cref{thm:maintheorem} in \cref{sec:precisestatement}. 
 
\begin{theorem}  \label{thm:maintheomintro}
	\cref{conj:linearanalogofh1} is true for Reissner--Nordström--AdS. 
\end{theorem}

\begin{rmk}
Since the existence of QNMs converging to the real axis expoentially fast as in \eqref{eq:qnms} also holds true for Kerr(--Newman)--AdS spacetimes, we expect that \cref{thm:maintheomintro} also extends to  Kerr--Newman--AdS spacetimes. 
\end{rmk}
\begin{rmk}
Our proof also extends to more general boundary condition such as Neumann or suitable Robin conditions for which no dissipation through $\mathcal I$ occurs.  
\end{rmk}

\paragraph{Brief overview of the proof} 
For the proof of \cref{thm:maintheomintro} and its more precise version \cref{thm:maintheorem}  (see already \cref{sec:precisestatement})  we will take an approach based on QNMs as introduced in \eqref{eq:qnms}. It suffices to take axisymmetric QNMs satisfying $m=0$ which we will do for simplicity.  As the QNMs are regular at the event horizon, we can and  will  extend them smoothly across the event horizon $\mathcal H_R$ as solutions to \eqref{eq:waveequation} on both the exterior and the interior.  Having the explicit separated form of \eqref{eq:qnms}, we then take the  limit   toward  the Cauchy horizon $\Ch_R$. The regularity of the solution at the Cauchy horizon will then depend on the behavior of the solution $R(r)$ to the radial o.d.e.\ coming from the separation of variables. To show $H^1$ blowup of $\psi$ we then show that the transmission coefficient $\mathfrak T$ does not have a zero at the quasinormal mode frequency. This is done via suitable approximations   of the radial o.d.e.\ solved by $R(r)$ in different domains in the black hole interior, see already \cref{eq:transmissioncoefnotzero}. This part of argument constitutes the most technical part. With the non-triviality of the transmission coefficient at hand, blowup of the local energy then follows directly for all QNMs which decay sufficiently slow on the exterior (i.e.\ which have sufficiently large angular momemtum $\ell$) such that  the bound (see already \eqref{eq:boundbeta})
\begin{align}\label{eq:boundbetaintro}
	\beta:= \frac{ -\operatorname{Im}(\omega_\ell)}{ \kappa_-} < \frac 12.
\end{align}
 is satisfied. Note that the condition \eqref{eq:boundbetaintro} has also played an crucial role in the $\Lambda>0$ case, see \cite{hintzvasyinterior,MR3697197,cardosoetal,dias2018strong,dias2018strong2,MR3882684}.
 \paragraph{A comment on the $\Lambda >0$ case.}
From the discussion above we recall that   \cref{thm:maintheomintro} is proved using QNMs and showing suitable lower bounds on the transmission coefficients in Fourier space. In contrast,  the linear analog of \cref{conj:linearanalogofh1} for $\Lambda =0$ is proved in  \cite{luk2016kerr} (see also \cite{dafermos2017time}) using a different physical-space approach in which lower bounds on  Price-law tails are propagated  into the interior via suitable  weighted energy estimates. 
 In the  remaining $\Lambda >0$ case, the linear analog  of \cref{conj:linearanalogofh1}  has no definitive resolution, yet: On the one hand, in the class of rough initial data ($H^{1+\epsilon}\times H^\epsilon$), the analog of \cref{conj:linearanalogofh1} for $\Lambda >0$ has been shown in \cite{dafermos2018rough}. On the other hand, in the class of smooth initial data ($C^\infty\times C^\infty$), the situation is different and it has been argued by considering QNMs that the analog of  \cref{conj:linearanalogofh1} for $\Lambda>0$ is false for Reissner--Nordström--de~Sitter close to extremality \cite{cardosoetal,dias2018strong,dias2018strong2,reallsurvey} but true for Kerr--de~Sitter \cite{MR3882684,MR4247555}. However, it still remains an open problem to put the analysis of   \cite{MR3882684,MR4247555} on a rigorous footing. We expect that the methods developed in the paper at hand---particularly the proof of the aforementioned lower bounds on the transmission coefficients in \cref{subsec:lowerboundont}---will turn out to be useful for such a proof.

 \paragraph{Outline of the paper} In \cref{sec:prelim21} we first set up the Reissner--Nordström--AdS family of black holes and recall the well-posedness of \eqref{eq:waveequation}. In \cref{sec:precisestatement} we state the main result \cref{thm:maintheorem} and recall the existence of quasinormal modes which converge to the real axis exponentially fast. Then, in \cref{sec:separationofvariablesandradialode} we recall the separation of variables, study the radial o.d.e.\ and show the lower bound on the transmission coefficients. Finally, in \cref{sec:proofofmaintheorem} we give the proof of \cref{thm:maintheorem} from which  \cref{thm:maintheomintro} directly follows.
 \paragraph{Acknowledgments} The author would like to express his gratitude to Mihalis Dafermos for many valuable comments on the manuscript and to Yakov Shlapentokh-Rothman for several useful discussions. This research is supported by Dr.\ Max Rössler, the Walter Haefner Foundation and the ETH Zürich Foundation.
 \section{Preliminaries}\label{sec:prelim21}
\subsection{The Reissner--Nordström--AdS black holes}
For black hole parameters $M>0, Q\neq 0, l^2\neq 0$ we define the polynomial  \begin{align}\Delta_{M,Q,l}(r):=r^2 - {2M}r + \frac{r^4}{l^2} + {Q^2}\end{align}
and define the non-degenerate set of subextremal parameters as 
\begin{align}
	\mathcal P :=\{ (M,Q,l) \in (0,\infty)\times \mathbb R \times (0,\infty) \colon \Delta_{M,Q,l}(r) \text{ has two postive roots satisfying } 0 < r_- < r_+  \}.
\end{align} 
We consider  \emph{\textbf{fixed}} parameters $M,Q,l,\alpha$, where \begin{align}\label{eq:parameters}(M,Q,l) \in \mathcal P \text{ and } \alpha < \frac 94 . \end{align}

Let the exterior  $\mathcal{R} $ be a smooth four dimensional manifold diffeomorphic to $\mathbb{R}^2 \times \mathbb{S}^2$.
On $\mathcal{R}$   we define---up to the well-known degeneracy of spherical coordinates---global coordinates \begin{align}\nonumber &(t  ,r, \theta ,\varphi  )  \in (r_+,\infty)\times \mathbb{R}\times \mathbb{S}^2.  \end{align}
 On the manifold $\mathcal{R}$ we define the Reissner--Nordström--Anti-de~Sitter metric
\begin{align}\label{eq:defnrnadsmetric}
	g := - \frac{\Delta(r)}{r^2}\d t\otimes\d t + \frac{r^2 }{\Delta}\d r \otimes\d r + r^2 (\d\theta \otimes\d\theta + \sin^2\theta \d \varphi\otimes \d\varphi ).
\end{align}

For $r\in (r_+,\infty)$ we  define the tortoise coordinate $r^\ast$ by
\begin{align}\label{eq:defnofrast}
	\frac{\d r^\ast}{\d r} := \frac{r^2}{\Delta},
\end{align}
with the normalization constant of $r^\ast(\infty ) = \frac{\pi }{2} l $. Then,  for $r>r_+$ we have \begin{align}r^\ast(r) =  \frac{1}{2 \kappa_+}\log(r - r_+) + f_\ast(r) \end{align}for some real-analytic \begin{align}f_\ast\colon (r_-, \infty) \to \mathbb R.\label{eq:firsttimefastarises}\end{align} Here, $\kappa_+$ is the surface gravity of the event horizon defined as 
\begin{align}
&\kappa_+ = \frac{1}{2} \partial_r \left( \frac{\Delta}{r^2} \right) \vert_{r=r_+} =\frac{r_+ -  M + 2r_+^3/l^2}{r_+^2}.
\end{align}

With $r^\ast$ at hand, we also define ingoing Eddington--Finkelstein type coordinates
\begin{align}
	v = t+r^\ast, r=r, \theta=\theta, \varphi= \varphi
\end{align}
in which the metric reads
\begin{align}
	g = - \frac{\Delta(r)}{r^2}\d v \otimes\d v +  2 \d v \otimes\d r + r^2 (\d\theta \otimes\d\theta + \sin^2\theta \d \varphi\otimes \d\varphi ).
\end{align}
Now, in these coordinates, $g$ extends smoothly to $(r,v,\theta,\varphi) \in  (r_-, \infty) \times \mathbb R \times \mathbb S^2 =: \mathring{ \mathcal M }$ and the region  $\mathcal R$ is identified with $r\in(r_+,\infty)$. We denote the region  $r\in (r_-,r_+)$ with    $\mathcal B$. The hypersurface $r=r_+$ is called the right event horizon $\mathcal H_R$.  We also define $T:= \partial_v$ which is a global Killing vector field. We also note that imposing $T$ to be future directed on $\mathcal R$ gives a unique time orientation on $(\mathring{\mathcal M},g)$. 

In the region $\mathcal B$, we   define $r^\ast$ as \begin{align}r^\ast (r) := \frac{1}{2\kappa_+}\log(r_+ - r) + f_\ast (r)\end{align}
which by construction satisfies  \eqref{eq:defnofrast}. We note that then we also have 
\begin{align}r^\ast(r) = \frac{-1}{2\kappa_-} \log(r - r_-) + g_\ast(r)\end{align} for some real-analytic $g_\ast \colon (0, r_+) \to \infty$.  Here, $\kappa_-$ is the surface gravity of the Cauchy horizon defined as 
\begin{align}
 	& \kappa_- = -\frac{1}{2} \partial_r \left( \frac{\Delta}{r^2} \right) \vert_{r=r_-} = -\frac{r_- -  M + 2r_-^3/l^2}{r_-^2}.
\end{align}
At this point we note that both $\kappa_+$ and $\kappa_-$ are positive as $\Delta$ is negative in $(r_-,r_+)$. We also remark that for $r\in (r_-,\infty)\setminus r_+$ we have 
\begin{align}
	|r-r_+|^{\frac{1}{2\kappa_+}} = e^{ r^\ast(r) - f_\ast(r) } 
\end{align}
and for $r\in (r_- ,r_+)$ we have 
\begin{align}
	|r-r_-|^{\frac{-1}{2\kappa_-}} = e^{ -r^\ast(r) + g_\ast(r) }.
\end{align}

In $\mathcal B$ we also set $t:= v-r^\ast$ and $u:=  2 r^\ast - v$. In particular, in the region $\mathcal B$ we   have coordinates $(t,r,\theta,\varphi)$ in which the metric takes  the form \eqref{eq:defnrnadsmetric}. Further, we will also use  outgoing Eddington--Finkelstein type coordinates $(r,u,\theta,\varphi)$ in $\mathcal B$ in which the metric takes the form
\begin{align}
	g = - \frac{\Delta(r)}{r^2}\d u \otimes\d u -  2 \d u \otimes\d r + r^2 (\d\theta \otimes\d\theta + \sin^2\theta \d \varphi\otimes \d\varphi ).
\end{align}
In these  outgoing Eddington--Finkelstein type coordinates $(r,u,\theta,\varphi)$  we attach the right  Cauchy horizon as $\mathcal{CH}_R = \{ r=r_- \}$.  The metric $g$ extends smoothly to $\mathcal M:= \mathring{\mathcal M} \cup \Ch_R$ and we have constructed the time-oriented Lorentzian manifold $(\mathcal M,g)$.  Further we write    $\mathcal I = \{r = \infty\}$  for the  conformally timelike boundary    at infinity and we refer to \cite{CKehle2019} for more details.

In order to pose our initial data for \eqref{eq:waveequation} we let \begin{align}\Sigma \subset \mathcal M\label{eq:defnofsigma}
\end{align} be a connected spacelike spherically symmetric hypersurface which extends to the conformal infinity $\mathcal I$ and transversally intersects the event horizon $\mathcal H_R$ with a boundary consisting of a single sphere in $\mathcal B$.  For simplicity we assume that for   $R \in \mathbb R$ sufficiently large,  $\Sigma$ agrees with the   $\{t=0\}$ hypersurface for $r\geq R$. 

We have illustrated the constructed spacetime as a Penrose diagram, the spacelike hypersurface $\Sigma$ as well as $D^+(\Sigma)$ in \cref{fig:rnads}.
   
\begin{figure}
	\begin{center}
		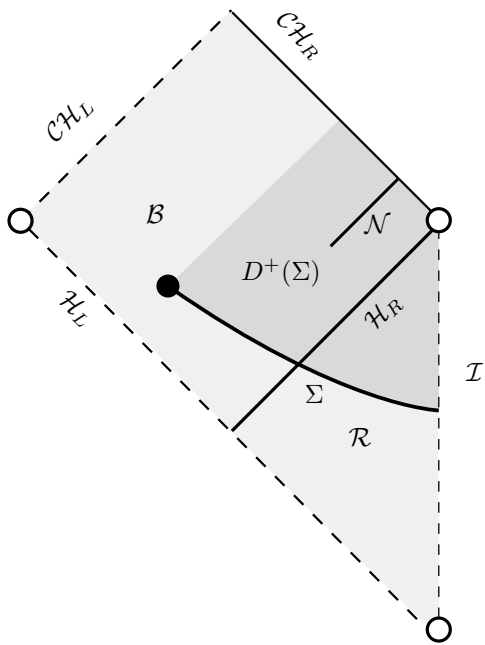
	\end{center}
	\caption{Penrose diagram of $(\mathcal M, g)$ and initial hypersurface $\Sigma$. The darker shaded region denotes $D^+(\Sigma)$.}
	\label{fig:rnads}
\end{figure}
 \subsection{Well-posedness of the Klein--Gordon equation \texorpdfstring{\eqref{eq:waveequation}}{}}
We will consider solutions to \eqref{eq:waveequation} with data posed on $\Sigma$ and Dirichlet conditions imposed on $\mathcal I$, where we say that a $C^1$ function $\psi$ on $\mathcal M$ satisfies the Dirichlet condition if 
\begin{align} \label{eq:dirichletcond}
	f^{-1}(r) \psi\to0 \text{ as }  r\to \infty,\end{align}
 where $f(r)$ is the twisting function defined as \begin{align} f(r):=  r^{-\frac 32+\sqrt{\frac 94 - \alpha}}. \end{align}

To do so we need further definitions. Let $\mathcal S$ be  an arbitrary connected spacelike hypersurface which extends to the conformal infinity $\mathcal I$. For any such $\mathcal S$ we  define   $n_{\mathcal S}$ as the unit future directed normal to $\mathcal S$ and the renormalized $\hat n_{\mathcal S}$ is then defined as \begin{align}\hat n_{\mathcal S} := r n_{\mathcal S}.\end{align}
To the twisting function $f$ we associate the twisted derivative as \begin{align}\tilde\nabla_\mu \psi:= f \nabla_\mu ( f^{-1} \psi).\end{align}
We further define the norms 
\begin{align}  \| \psi \|_{\underline L^2(\mathcal S)}^2 := \int_{\mathcal S} |\psi|^2 r^{-1} \dvol_{\mathcal S} \text{ and } \| \psi \|_{\underline H^1 (\mathcal S)}^2 := \int_{\mathcal S} \left(|\tilde \nabla \psi|^2 + |\psi|^2 r^{-2} \right) r \dvol_{\mathcal S}.\end{align} Finally, we define $\underline H_0^1(\mathcal S)$ as the completion of smooth compactly supported functions on $\mathcal S$ with respect to the norm $\| \cdot \|_{\underline H^1(\mathcal S)}$. 
As  $\Sigma$ agrees with a constant $t$-hypersurface for $r\geq R$, we note that 
\begin{align}
\dvol_{\Sigma} = \frac{r^3}{\sqrt{\Delta}} \d r \d \sigma_{\mathbb S^2} \sim r\d r \d \sigma_{\mathbb S^2} 
\end{align}
for $r\geq R$ and thus, 
\begin{align}& |\psi|^2 r^{-1} \dvol_{\Sigma} \sim |\psi|^2 \d r \d \sigma_{\mathbb S^2}, \\
	&|\tilde \nabla \psi |^2 r \dvol_{\Sigma} \sim |\tilde \nabla \psi|^2 r^2 \d r \d \sigma_{\mathbb S^2} \end{align} for $r\geq R$. 
We now state the following well-posedness as established in \cite{warnick_massive_wave,wellposed}. 
\begin{prop}[\cite{warnick_massive_wave,wellposed}]\label{prop:wellposed}
Let $\Psi_0 \in \underline H_0^{1}(\Sigma), \Psi_1 \in \underline L^2(\Sigma)$. Then, there exists a unique (weak) solution $\psi$ to \eqref{eq:waveequation} on $  D^+(\Sigma)$  such that $\psi\vert_\Sigma = \Psi_0$, $\hat n_{\Sigma} \psi \vert_{\Sigma} = \Psi_1$ with Dirichlet boundary conditions imposed at $\mathcal I$. Moreover, if $\mathcal S\subset   D^+(\Sigma)$ is a spacelike hypersurface, then $\psi\vert_{\mathcal S} \in \underline H_0^{1}(\mathcal S), \hat n_{\mathcal S} \psi \vert_{\mathcal S} \in \underline L^2(\mathcal S)$. 
\end{prop}
\begin{rmk}
	Similar statements as in \cref{prop:wellposed} hold for higher regularity replacing $\underline H_0^{1}(\Sigma),  \underline L^2(\Sigma)$ with appropriate $\underline H_0^{k+1}(\Sigma),  \underline H^{k}(\Sigma)$, see \cite[Section~5.3]{warnick_massive_wave}. Thus, for sufficiently regular initial data we obtain classical solutions to the initial boundary value problem.
\end{rmk}

\paragraph{Conventions}
With $a\lesssim b$ for $a\in \mathbb{R}$ and $b\geq 0$ we mean that there exists a constant $C(M,Q,l,\alpha, \Sigma) > 0$ with $a\leq C b$. If $C(M,Q,l,\alpha, \Sigma)$ depends on an additional parameter, say $\ell$, we will write $a \lesssim_\ell b$.  We also use $a\sim b$ for some $a,b\geq 0$ if there exist constants $C_1(M,Q,l,\alpha,\Sigma), C_2(M,Q,l,\alpha,\Sigma) >0$ with $C_1 a \leq b \leq C_2 a$. 
We shall also make use of the standard Landau notation $O$ and $o$ \cite{olver}. To be more precise, let $X$ be a point set (e.g.\ $X=\mathbb R, [a,b], \mathbb C$)  with limit point $c$. As $x\to c$ in $X$, $f(x) = O(g(x))$ (resp.\ $f(x) = o(g(x))$) means $\frac{|f(x)|}{|g(x)|} \leq C(M,Q,l,\alpha,\Sigma)$ (resp.\ $\frac{|f(x)|}{|g(x)|} \to 0$ with rate depending on $M,Q,l,\alpha,\Sigma$). We write $O_\ell(g(x))$ or $o_\ell(g(x))$ if the constant $C$ or the rate depend  on an additional parameter $\ell$.

\section{Precise statement of the main result}
\label{sec:precisestatement}
 We now give the precise verison of our main result \cref{thm:maintheomintro}.
 
\begin{theorem}\label{thm:maintheorem}
	There exist smooth Cauchy data $(\Psi_0,\Psi_1) \in C^\infty(\Sigma)\times C^\infty(\Sigma)\cap \underline H_0^1(\Sigma) \times \underline L^2(\Sigma)$ for \eqref{eq:waveequation} such that the corresponding smooth solution $\psi$ with Dirichlet boundary conditions imposed at $\mathcal I$ satisfies 
	\begin{align}\label{eq:blowup}
		\| \psi\|_{\dot{H}^1(\mathcal N)}  = \infty,
	\end{align}
	i.e.\ $\psi \notin H^1(\mathcal N)$, where $\mathcal N$ is any constant $u$ hypersurface emanting from a sphere in $D^+(\Sigma)\cap  \mathcal B $ and terminating on a sphere of $\mathcal{CH}_R$. 
	
Moreover, blowup of the local energy as in \eqref{eq:blowup} is a generic property of the Cauchy data in the following sense. The space $H$ of smooth Cauchy data in  $C^\infty(\Sigma)\times C^\infty(\Sigma)\cap \underline H_0^1(\Sigma) \times \underline L^2(\Sigma)$ for which \eqref{eq:blowup} is not true is exceptional in the sense that $H$ has infinite co-dimension in  $C^\infty(\Sigma)\times C^\infty(\Sigma)\cap \underline H_0^1(\Sigma) \times \underline L^2(\Sigma)$.
\end{theorem} 

We   note that \cref{thm:maintheorem} implies \cref{thm:maintheomintro}. 
\begin{rmk} The statement \eqref{eq:blowup} is independent of the choice of volume form   on $\mathcal N$. In particular, for any fixed choice of the volume form on $\mathcal N$ we have 
\begin{align}\label{eq:h1dotnorm} \| \psi\|_{\dot{H}^1(\mathcal N)}^2 \sim \int_{\mathcal N} \left( |\partial_r \psi|^2(r,u,\theta,\varphi) + |\slashed \nabla \psi|^2 (r,u,\theta,\varphi)\right)r^2  \d r \d\sigma_{\mathbb S^2},\end{align}
where the implicit constant in $\sim$   depends on the fixed choice of volume form on $\mathcal N$ and with  $\slashed \nabla$ we denote the angular derivatives.
\end{rmk}

The proof of \cref{thm:maintheorem}  will be given in \cref{sec:proofofmaintheorem} and is based on an explicit quasinormal mode construction and will crucially rely on the existence of quasinormal modes on the exterior with imaginary part converging exponentially fast to the real axis. This has been established for Schwarzschild--AdS and Kerr--AdS  in \cite{quasinormal_gannot,MR3223487,quasimodes}, see also \cite[Theorem~11]{gannot2016quasinormal}.  Mutatis mutandis the proof of the existence of quasinormal modes in \cite[Theorem~11]{gannot2016quasinormal} for Schwarzschild--AdS also applies to Reissner--Nordström--AdS and we obtain the following theorem.

\begin{theorem} [\cite{MR3223487,gannot2016quasinormal}]
	\label{thm:existenceofqnm}
There exists a sequence of quasinormal mode frequencies $\omega_\ell \in \mathbb C$, $\ell \geq \ell_0$ for some $\ell_0$ sufficiently large and  non-trivial quasinormal modes $R_\ell  \in C^\infty(r_+, \infty)$ such that 
\begin{align}
	\psi_\ell (t, r, \theta,\varphi)= e^{-i \omega_\ell t} R_\ell(r) Y_{\ell,0}(\theta) 
\end{align}
satisfies \eqref{eq:waveequation} on $\mathcal R$ and
\begin{align} \label{eq:reomegal}
 &  \operatorname{Re}(\omega_{\ell} ) \sim \ell,  \\
&0 < - \operatorname{Im}(\omega_{\ell} )  \lesssim    e^{-c_{\mathrm{I}} \ell } \label{eq:imomegal}
\end{align}
for some constant $c_I >0$. 
Moreover, $R_\ell$ satisfies the   boundary conditions 
\begin{enumerate}
	\item \label{item1} $  R_\ell(r) =     (r - r_+)^{\frac{-i \omega}{2 \kappa_+}} \tilde  R_{\ell} ( r)    $, where  $\tilde R_{\ell} ( r)$ extends smoothly to $r\in (r_-,  \infty)$ with $\tilde R_{\ell} ( r_+)  = 1$. 
	\item  \label{item2}$ \lim_{r\to\infty}   r^{\frac 32-\sqrt{\frac 94 - \alpha}}  R_\ell( r)   =0$.
\end{enumerate}
\end{theorem}
\begin{rmk}
Note that  boundary condition~\ref{item1} in \cref{thm:existenceofqnm} corresponds to the fact that  $\psi_\ell$ extends smoothly  across the event horizon $\mathcal H_R$ and  boundary condition~\ref{item2}  corresponds to the fact that $\psi_\ell$ satisfies the Dirichlet boundary conditions at $\mathcal I$ as in \eqref{eq:dirichletcond}.
\end{rmk}
\begin{rmk}
Strictly speaking, the existence of quasinormal modes  has been only shown for Schwarzschild--AdS and Kerr--AdS in \cite{MR3223487,quasinormal_gannot,quasimodes}, but---mutatis mutandis---the proof directly extends to Reissner--Nordström--AdS. 
\end{rmk}
\section{Separation of variables and radial o.d.e.}
\label{sec:separationofvariablesandradialode}
We begin by reviewing the procedure of separation of variables for solutions to \eqref{eq:waveequation}. Formally, we have that \eqref{eq:waveequation} admits mode solutions  of the form
\begin{align}
\psi(t,r,\theta,\varphi) = e^{-i\omega t } R( \omega,\ell, r) Y_{\ell, m}(\theta) e^{i m \varphi},
\end{align}
for $\omega \in \mathbb C$, $\ell \in \mathbb N_0$, $m \in \mathbb Z_{|m|\leq \ell}$. Here $Y_{\ell, m } (\theta)e^{im \varphi} $ denote  the standard spherical harmonics. For the purpose of proving \cref{thm:maintheorem} it is sufficient to consider only axisymmetric mode solutions ($m=0$), i.e. solutions of the form
\begin{align}\label{eq:separated}
\psi(t,r,\theta) = e^{-i\omega t } R(\omega,\ell,r)  Y_{\ell,0}(\theta)
\end{align}
for $\omega \in \mathbb C$, $\ell \in \mathbb N_{0}$. We note that a mode solution $\psi(t,r,\theta)$ of the form \eqref{eq:separated} solves \eqref{eq:waveequation} on $\mathcal R$ (resp.\ on $\mathcal B$)  if and only if $R(\omega,\ell,r)$ satisfies the radial o.d.e.\ 
\begin{align}\label{eq:radialode1}
 \frac{\Delta}{r^4} \frac{\d }{\d r } \left(\Delta  \frac{\d R}{\d r }\right) + \omega^2 R -\ell (\ell+1)\frac{\Delta}{ r^4}  R + \frac{\alpha}{l^2}\frac{\Delta}{r^2} R= 0 
\end{align}
  for $r\in (r_+,\infty)$ (resp.\ for  $r\in (r_-, r_+)$). 
We now define fundamental solutions to \eqref{eq:radialode1} associated to the singular points of \eqref{eq:radialode1} located at $r=r_-, r_+, \infty$. 
\subsection{Fundamental solutions associated to the singular points}
\begin{lemma}\label{lem:defnrs}
	For $\omega \in \mathbb C $ with  $|\operatorname{Im}(\omega) | < \min(\kappa_+,  \kappa_- )$, $\ell \in \mathbb Z_{\geq 0} $ there exist solutions $R_{D}, R_{N}, R_{\mathcal H_R}$, $R_{\mathcal{CH}_L}$, $R_{\mathcal{CH}_R}$ to \eqref{eq:radialode1} with the following properties. 
		\begin{align}\label{eq:rd}
		& R_{D}( \omega, \ell, r) = r^{-\frac 32- \sqrt{\frac 94  - \alpha} } \tilde R_D( \omega, \ell, r) \text{ on }  r \in(r_+,\infty), \\	
	\label{eq:rn}	& R_{N} ( \omega, \ell, r)= r^{-\frac 32 + \sqrt{\frac 94  - \alpha} } \tilde R_N( \omega, \ell, r) + C_3  (\omega,\ell)  r^{-\frac 32 - \sqrt{\frac 94 -\alpha}  } \log \frac 1 r\text{ on }  r \in(r_+,\infty), \\
		&	\label{eq:rhr1} R_{\mathcal H_{R}} ( \omega, \ell, r) = |r-r_+|^{- \frac{i \omega}{2\kappa_+}} \tilde R_{\mathcal H_R} ( \omega, \ell, r) = e^{-i\omega r^\ast(r) + i \omega f_\ast(r)}  \tilde  R_{\mathcal H_R} (\omega,\ell,r) \text{ on }  r\in (r_-, \infty) \setminus \{ r_+ \},\\
		&		R_{\mathcal{CH}_L} ( \omega, \ell, r) = (r-r_-)^{ - \frac{i \omega}{2\kappa_-}} \tilde R_{\mathcal{CH}_L}( \omega, \ell, r) =  e^{-i \omega r^\ast(r) + i \omega g_\ast(r) }  \tilde R_{\mathcal{CH}_L}  (\omega,\ell,r)  \text{ on }  r\in (r_-, r_+)   ,\\
		&			R_{\mathcal{CH}_R} ( \omega, \ell, r) = (r-r_-)^{ \frac{i \omega}{2\kappa_-}} \tilde R_{\mathcal{CH}_R}( \omega, \ell, r)=   e^{ i \omega r^\ast(r) - i \omega g_\ast(r) } \tilde  R_{\mathcal{CH}_R} (\omega,\ell,r) \text{ on }  r\in (r_-, r_+ )  ,\label{eq:rchl}
	\end{align}
	where \begin{itemize}
		\item $R_N$ and $R_D$ are linearly independent,
		\item $\tilde R_{D}, \tilde R_{N}$ are analytic for $r\in (r_+, \infty)$ with $\lim_{r\to\infty} \tilde R_{D}(\omega,\ell,r) = \lim_{r\to\infty} \tilde R_{N}(\omega,\ell,r)= 1$,
		\item $\tilde R_{D}(\omega,\ell, 1/s),\tilde R_{N}(\omega,\ell, 1/s) $ extend analytically across $s=0$ ($r=\infty$),
		\item $C_3=0$ if $\alpha \notin \mathcal E := \{ \frac{ 1}{4}(9-k^2)\colon k \in \mathbb N\}$,
		\item  $  \tilde   R_{\mathcal H_R}$ is analytic for $r\in (r_-, \infty)$ satisfying $\tilde R_{\mathcal H_R}(\omega, \ell, r_+) = 1$,
		\item   $\tilde R_{\mathcal{CH}_R}, \tilde R_{\mathcal{CH}_L}$ are analytic for $r\in [r_-,r_+) $ and satisfy $\tilde R_{\mathcal{CH}_R}(\omega, \ell, r_-)=\tilde R_{\mathcal{CH}_L}(\omega, \ell, r_-)  =1.$
	\end{itemize}
\end{lemma}
\begin{proof}
	The points $r\in \{ r_-, r_+, \infty\}$ constitute regular singular points of the radial o.d.e.\ \eqref{eq:radialode1}. A standard indicial analysis (see \cite[Section 2.2]{dold}) around the  points show  \eqref{eq:rd}--\eqref{eq:rchl}.   Note that the bound on $\operatorname{Im}(\omega)$ guarantees that the indicial roots at $r_+$ and $r_-$ do not differ by an integer. Thus, for the singular points $r_-$ and $r_+$, logarithmic terms do not occur. For more details we refer to \cite[Section 2.2]{dold}, see also  \cite{srgrwoingmodes,olver2014asymptotics}.
 \end{proof}
\subsection{The re-normalized radial o.d.e.\ and the scattering coefficients}
It is also useful to   write the radial o.d.e.\ \eqref{eq:radialode1} in the coordinate $r^\ast$. It will then take a standard Schrödinger form. 

More precisely,  $R(\omega,\ell,r) $ solves \eqref{eq:radialode1} on $\mathcal B$ (and respectively on $\mathcal R$) if and only if 
\begin{align}u(\omega,\ell,r^\ast) := r(r^\ast) R(\omega,\ell,r(r^\ast)) \end{align}
solves 
\begin{align}\label{eq:radialodeforu}
	- u'' + 	(V_\ell - \omega^2) =0,
\end{align}
where $' := \frac{\d }{\d r^\ast}$. The potential $V_\ell(r(r^\ast))$ is given by 
 \begin{align}V_\ell (r)= h(r)\left( r^{-1}  \frac{\d h(r) }{\d r} + \frac{\ell (\ell +1) }{r^2} - \frac{\alpha}{l^2}\right),
 \end{align}  where 
\begin{align}h(r):= \frac{\Delta(r) }{r^2}.\end{align} 

On $\mathcal B$ and for $\ell\in \mathbb Z_{\geq 0}$, $\omega \in \mathbb C  $ with $|\operatorname{Im}(\omega) |< \min(\kappa_+, \kappa_-)$ we define solutions to \eqref{eq:radialodeforu} as
\begin{align}\label{eq:defnuhra}
	&\uhr (\omega,\ell, r^\ast):=   \frac{r(r^\ast)}{r_+ e^{i \omega f_\ast(r_+)}}  R_{\mathcal H_R} (\omega,\ell,r(r^\ast)) \\ \label{eq:defnuchra}
	&\uchr (\omega,\ell, r^\ast):=   \frac{r(r^\ast)}{r_- e^{-i \omega g_\ast(r_-)}}   R_{\mathcal{CH}_R}  (\omega,\ell,r(r^\ast)) \\ \label{eq:defnuchla}
	&\uchl  (\omega,\ell, r^\ast):=  \frac{ r(r^\ast)}{r_- e^{i \omega g_\ast(r_-)}}   R_{\mathcal{CH}_L} ((\omega,\ell,r(r^\ast)) 
\end{align}
such that they have the asymptotics
\begin{align}\label{eq:asymptoticsuhr}
	&\uhr (\omega,\ell, r^\ast) = e^{-i\omega r^\ast} (1 + o_{\ell,\omega}(1)) \text{ as } r^\ast \to -\infty,   \\ \label{eq:asymptoticscuhr}
	&\uchr (\omega,\ell, r^\ast) = e^{i\omega r^\ast} (1 + o_{\ell,\omega}(1)) \text{ as } r^\ast \to  \infty,\\
	&\uchl  (\omega,\ell, r^\ast) = e^{-i\omega r^\ast} (1 + o_{\ell,\omega}(1))  \text{ as } r^\ast \to  \infty.
\end{align}
We   also define that the Wronskian  \begin{align}
	\mathfrak W[u_1,u_2](\omega,\ell,r^\ast) := u_1(\omega,\ell,r^\ast)  u_2' (\omega,\ell,r^\ast) - u_1'(\omega,\ell,r^\ast)  u_2(\omega,\ell,r^\ast)   \end{align}
 and note that  it does not depend on $r^\ast$ for two solutions $u_1,u_2$ of \eqref{eq:radialodeforu}, i.e.\
 \begin{align}
\frac{\d}{\d r^\ast} \mathfrak W[u_1,u_2](\omega,\ell,r^\ast)  =0
 \end{align}
for  two solutions $u_1,u_2$. Moreover, two solutions $u_1,u_2$ are linearly dependent if and only if their Wronskian vanishes, i.e.\ if and only if  $\mathfrak W[u_1,u_2](\omega,\ell)=0$.

\begin{lemma}\label{lem:defnofRandT}
	For $\omega \in \mathbb C\setminus \{ 0 \}$ with  $|\operatorname{Im}(\omega) | < \min(\kappa_+, \kappa_-)$, $\ell \in \mathbb Z_{\geq 0} $ there exist unique real-analytic transmission and reflection coefficients $\mathfrak T(\omega,\ell) $ and $\mathfrak R(\omega,\ell)$ such that
	\begin{align}
	&\uhr(\omega, \ell, r^\ast) = \mathfrak R (\omega,\ell) \uchr (\omega, \ell, r^\ast)+ \mathfrak T (\omega,\ell) \uchl (\omega, \ell, r^\ast). \label{eq:defnoftandr2}
	\end{align}
Equivalently, 
\begin{align}\label{eq:twronskian}
& \mathfrak T(\omega,\ell):= \frac{ \mathfrak W[\uhr, \uchr](\omega,\ell)}{\mathfrak W[\uchl, \uchr](\omega,\ell)}= \frac{ \mathfrak W[\uhr, \uchr](\omega,\ell)}{2i \omega } ,\\
&\mathfrak R(\omega,\ell):= \frac{ \mathfrak W[\uhr, \uchl](\omega,\ell)}{\mathfrak W[\uchr, \uchl](\omega,\ell)}= \frac{ \mathfrak W[\uhr, \uchl](\omega,\ell)}{-2i \omega }.\label{eq:rwronskian}
\end{align}
\begin{proof}
For $\omega \neq 0$, the solutions  $\uchl $ and $\uchr$   are linearly independent and   form a fundamental pair. Thus, there exist unique coefficients $\mathfrak T(\omega,\ell) $ and $\mathfrak R(\omega,\ell)$ such that \eqref{eq:defnoftandr2} holds. The   identities \eqref{eq:twronskian}, \eqref{eq:rwronskian} follow now from applying the Wronskian to \eqref{eq:defnoftandr2}.
\end{proof}
\end{lemma}
 \subsection{Lower bound on transmission coefficient}
 \label{subsec:lowerboundont}
 The main technical ingredient of the proof of \cref{thm:maintheorem} is a lower bound on the transmission coefficient $\mathfrak T(\omega,\ell)$ at the quasinormal mode frequencies $\omega=\omega_\ell$. This is proved in the following via suitable approximations of solutions to \eqref{eq:radialodeforu}.
\begin{lemma}\label{eq:transmissioncoefnotzero}
The transmission coefficient satisfies \begin{align}\label{eq:lowerboundontT}
|\mathfrak T(\omega = \omega_\ell,\ell)| \gtrsim 1 
\end{align}
for $\ell$ sufficiently large, where $\omega_\ell$ are as in \cref{thm:existenceofqnm} satisfying \eqref{eq:reomegal} and \eqref{eq:imomegal}.
\begin{proof}
We recall the radial o.d.e.\ in the interior 
\begin{align}\label{eq:recallradialode}
- u'' (\omega_\ell,\ell,r^\ast)   + (V_\ell (r^\ast)-\omega_\ell^2) u (\omega_\ell,\ell,r^\ast) = 0,
\end{align}
where  
 $\operatorname{Re}(\omega_\ell) \sim \ell$ and $0<-\operatorname{Im}(\omega_\ell) \lesssim e^{-c_\textup{I} \ell }$ from  \cref{thm:existenceofqnm}. In the following we will estimate the absolute value of the Wronskian $
|  \mathfrak W[u_{\mathcal{H}_R} , u_{\mathcal{CH}_R} ](\omega_\ell,\ell)|$ from below.

In view of the asymptotics of $u_{\mathcal H_R}$ and $u_{\mathcal{CH}_R}$ in \eqref{eq:asymptoticsuhr} and \eqref{eq:asymptoticscuhr}, we can equivalently define $u_{\mathcal H_R} (\omega_\ell,\ell,r^\ast) $ and $u_{\mathcal{CH}_R}(\omega_\ell,\ell,r^\ast) $ to be the unique solutions   to the Volterra integral equations 
\begin{align}\label{eq:defnuhr}
u_{\mathcal H_R} (\omega_\ell,\ell,r^\ast)= e^{-i \omega_\ell r^\ast} + \int_{-\infty}^{r^\ast} \frac{\sin(\omega_\ell (r^\ast - y) )}{\omega_\ell} V_\ell(y) u_{\mathcal H_R} (\omega_\ell, \ell, y) \d y, \\
u_{\mathcal{CH}_R} (\omega_\ell,\ell,r^\ast)   = e^{i \omega_\ell r^\ast} - \int_{r^\ast}^{\infty }\frac{\sin(\omega_\ell (r^\ast - y) )}{\omega_\ell} V_\ell(y) u_{\mathcal{CH}_R} (\omega_\ell, \ell, y) \d y\label{eq:defnuchr}
\end{align}
which converge for $\ell$ sufficiently large (i.e.\ $|\operatorname{Im}(\omega_\ell)|$ sufficiently small) as the potential satisfies 
\begin{align}|\label{eq:estonvell}
	V_\ell(y)| \lesssim \ell^2 e^{- 2  \min(\kappa_+,  \kappa_- ) |y| }.
\end{align}
That indeed  solutions to \eqref{eq:defnuhr} and \eqref{eq:defnuchr} exist follows from the standard theory of Volterra equations (e.g.\ \cite[Chapter~6, \S10]{olver2014asymptotics} or \cite[Proposition 2.3]{kehle2018scattering}) and the estimate \eqref{eq:estonvell}; see also the quantitative analysis in Region I below.

 In view of $\operatorname{Im}(\omega_\ell)  < 0$, note that $u_{\mathcal H_R}$ and $u_{\mathcal{CH}_R}$ blow up  as $r^\ast \to -\infty$ and $r^\ast \to \infty$, respectively. 

We will consider three different regions which are separated by 
\begin{align}R_1:= \frac{1}{\operatorname{Im({\omega_\ell})}}, R_2 := -  \frac{1}{\operatorname{Im({\omega_\ell})}}.\end{align}
\paragraph{Region I: $-\infty< r^\ast \leq R_1$}
In this region we define 
\begin{align}v_{\mathcal H_R } (\omega_\ell,\ell,r^\ast) := e^{i \omega_\ell r^\ast} \uhr(\omega_\ell,\ell,r^\ast) -1\end{align} and note that $v_{\mathcal H_R}$ solves 
\begin{align}\nonumber
 v_{\mathcal H_R} (\omega_\ell,\ell,r^\ast)=  & \int_{-\infty}^{r^\ast} \frac{\sin(\omega_\ell (r^\ast - y) )}{\omega_\ell} e^{i \omega_\ell (r^\ast - y)} V_\ell (y)  \d y \\
 & +  \int_{-\infty}^{r^\ast} \frac{\sin(\omega_\ell (r^\ast - y) )}{\omega_\ell} e^{i \omega_\ell (r^\ast - y)} V_\ell(y) v_{\mathcal H_R} (\omega_\ell, \ell, y) \d y. \label{eq:vhrvolterra}
\end{align}
For the kernel we have the estimate 
\begin{align} \nonumber
\int_{-\infty}^{R_1}  \sup_{r^\ast \in (y, R_1)} \left| \frac{\sin(\omega_\ell( r^\ast - y) )  }{\omega_\ell} e^{i \omega_\ell (r^\ast -y)} V_\ell(y) \right|   \d y &  \lesssim \int_{-\infty}^{R_1}  e^{-2\operatorname{Im}(\omega_\ell) (R_1 -y) } | \omega_\ell^{-1} | \ell^2 e^{   2  \min(\kappa_+,  \kappa_- )  y }  \d y \\
\lesssim \int_{-\infty}^{R_1}    \ell e^{ 2(   \operatorname{Im}(\omega_\ell) +  \min(\kappa_+,  \kappa_- )  ) y } &  \d y 
  \lesssim    e^{ \min(\kappa_+,  \kappa_- ) R_1} \lesssim    e^{- \tilde c  e^{c_{\operatorname{I} \ell } }} \lesssim e^{- \tilde c \ell} 
\end{align}
for some $\tilde c >0$ only depending on the black hole parameters. Here, we also used that \begin{align}\ell e^{ \min(\kappa_+,  \kappa_- ) R_1} \lesssim  \ell e^{-  \tilde c  e^{c_I\ell} }\lesssim 1.\end{align} 
Then, from standard estimates on Volterra integral equations (e.g.\ \cite[Chapter~6, \S10]{olver2014asymptotics} or \cite[Proposition 2.3]{kehle2018scattering}) we obtain 
\begin{align} \label{eq:v1}
&\|v_{\mathcal H_R}\|_{L^\infty( -\infty, R_1)}(\omega_\ell,\ell) \lesssim e^{- \tilde c \ell} ,\\
&\|v'_{\mathcal H_R}\|_{L^\infty(-\infty, R_1)} (\omega_\ell,\ell) \lesssim e^{- \tilde c \ell} . \label{eq:v2} 
\end{align}
\paragraph{Region II: $R_1 \leq r^\ast \leq R_2$}
In this region we use a WKB approximation. We define the complex potential  \begin{align}f(r^\ast):=   V_\ell(r^\ast) - \omega_\ell^2 =  V_\ell(r^\ast) - \operatorname{Re}(\omega_\ell)^2 + \operatorname{Im}(\omega_\ell)^2 - 2i \operatorname{Re}(\omega_\ell) \operatorname{Im} (\omega_\ell) .\end{align}
For the fractional power we choose the principal branch with $\arg \in (-\pi, \pi]$. 
We remark that $\operatorname{Im} (f)  > 0$ as $\operatorname{Re}(\omega_\ell)> 0 $ and $ \operatorname{Im}(\omega_\ell) <0$ and thus, $\operatorname{Re} ( \sqrt{f} ), \operatorname{Im} ( \sqrt{f} ) >0$. We also note that $\mathbb R\ni r^\ast \mapsto f(r^\ast)$ is real analytic as   $V(r)$  and $r(r^\ast)$ are real analytic. Thus, $f(r^\ast)$ extends holomorphically to a complex neighborhood  of the real line on which $f$ manifestly does not vanish. 
We will now apply standard results on WKB approximation with complex potentials (e.g.\   \cite[Chapter~6, Theorem 11.1]{olver2014asymptotics} or \cite[p.~29]{bsbook}). Writing \begin{align}\nonumber 
	\int_{R_1}^{r^\ast} \sqrt{f(y)  } \d y & =  i \omega_\ell (  r^\ast -R_1) + i \omega_\ell   \int_{R_1}^{r^\ast}  \sqrt{ 1 - \frac{V_\ell (y)}{\omega_\ell^2} } -1 \d y \\ & =  -i \omega_\ell R_1 +  i \omega_\ell   r^\ast + i \omega_\ell   \int_{R_1}^{r^\ast}  \sqrt{ 1 - \frac{V_\ell (y)}{\omega_\ell^2} } -1 \d y\end{align}
we use  \cite[Chapter~6, Theorem 11.1]{olver2014asymptotics}  to obtain two linear independent solutions  with reference points $R_2$ for $u_{WKB_1}$ and $R_1$ for  $u_{WKB_2} $ which are of the form
\begin{align}\label{eq:approxofuhr}
&u_{WKB_1} (\omega_\ell, \ell, r^\ast)= \frac{\sqrt{\omega_\ell}  }{f^{\frac 14} (r^\ast) } e^{ - i  \omega_\ell r^\ast} e^{-i \omega_\ell \int_{R_1}^{r^\ast} \left( \sqrt{1-\frac{V_\ell(y)}{\omega_\ell^2} }-  1 \right)  \d y } (1 + \eta_{WKB_1}  (\omega,\ell, \ell, r^\ast) ),\\
& u_{WKB_2}  (\omega_\ell, \ell, r^\ast)= \frac{\sqrt{\omega_\ell} }{f^{\frac 14}  (r^\ast)} e^{  i  \omega_\ell r^\ast} e^{ i \omega_\ell \int_{R_1}^{r^\ast} \left( \sqrt{1-\frac{V_\ell(y)}{\omega_\ell^2} }-  1 \right)  \d y } (1 + \eta_{WKB_2}  (\omega_\ell, \ell, r^\ast) )
\end{align}
where 
 \begin{align}&\eta_{WKB_1} (\omega_\ell, \ell, R_2)  = \eta'_{WKB_1}(\omega_\ell, \ell , R_2) =0, \\
 	& \eta_{WKB_2}(\omega_\ell, \ell, R_1)  = \eta'_{WKB_2}(\omega_\ell, \ell, R_1)    =0\end{align}  
 and
\begin{align} & \sup_{r^\ast \in (R_1,R_2)}( | \eta_{WKB_1} | + | f^{-\frac 12 } \eta'_{WKB_1}| )(\omega_\ell, \ell)\lesssim \exp(|F|(\omega_\ell, \ell)) -1,\\
& \sup_{r^\ast \in (R_1,R_2)} ( | \eta_{WKB_1} | + | f^{-\frac 12 } \eta'_{WKB_1}| )(\omega_\ell, \ell) \lesssim \exp(|F|(\omega_\ell, \ell)) -1 \end{align}  for 
\begin{align}F (\omega_\ell, \ell) = \int_{R_1}^{R_2}   \left| \frac{1}{f^{\frac 14}}\frac{\d^2}{\d {r^\ast}^2}  	f^{-\frac 14}  \right| \d r^\ast.
\end{align}

 We shall note that the reference points $R_2$ for  $u_{WKB_1} $ and the reference point $R_1$ for  $u_{WKB_2} $ are valid as $\operatorname{Re}	\int_{R_1}^{\tilde r^\ast} \sqrt{f(y)  } \d y$ is nondecreasing from $R_1$ to $r^\ast$ and vice versa $\operatorname{Re}	\int_{R_1}^{\tilde r^\ast} \sqrt{f(y)  } \d y$ is nonincreasing  from $R_2$ to $r^\ast$, i.e.\ \cite[Chapter~6, Theorem 11.1, conditions (i), (ii)]{olver2014asymptotics} are fulfilled. 
 
 Now, we use \begin{align} \label{eq:decayforq} 
&	\left|\frac{\d}{\d r^\ast } f(r^\ast)\right|, 	\left|\frac{\d^2}{\d {r^\ast}^2 } f(r^\ast)\right|\lesssim\ell^2  e^{- 2  \min(\kappa_+,  \kappa_- ) |r^\ast|},  \text{ and } 	|f^{\frac 14}| \gtrsim \ell^{\frac 12} 
\end{align}
 to obtain the quantitative bound 
\begin{align}
 \int_{R_1}^{R_2}  \left| \frac{1}{f^{\frac 14}}\frac{\d^2}{\d {r^\ast}^2}  	f^{-\frac 14} \right|  \d r^\ast \lesssim \frac{1}{\ell}
\end{align}
and hence,
\begin{align}
 \sup_{r^\ast \in [R_1,R_2]}( |\eta_{WKB_2}| + | f^{-\frac 12 }\eta'_{WKB_2} | + |\eta_{WKB_1}| + | f^{-\frac 12 }\eta'_{WKB_1} |)  (\omega_\ell, \ell) \lesssim \frac{1}{\ell}. \label{eq:estimatesoneta}
\end{align}
In particular, this shows that 
\begin{align}|\mathfrak W[u_{WKB_1},u_{WKB_2}]|(\omega_\ell, \ell) \sim |\omega_\ell| \sim \ell \end{align} by evaluating the Wronskian at $R_1$ and using the bounds of \eqref{eq:estimatesoneta}.

We further decompose 
\begin{align}\uhr (\omega_\ell,\ell,r^\ast) = A_{WKB_1} (\omega_\ell, \ell) u_{WKB_1}  (\omega_\ell,\ell,r^\ast)+ B_{WKB_2}(\omega_\ell, \ell) u_{WKB_2} (\omega_\ell,\ell,r^\ast)
\label{eq:expansionofuhr}
\end{align}
for \begin{align}
& A_{WKB_1}  (\omega_\ell, \ell)= \frac{ \mathfrak W [\uhr, u_{WKB_2} ] (\omega_\ell, \ell)}{\mathfrak W[u_{WKB_1},u_{WKB_2}](\omega_\ell, \ell)}\\
& B_{WKB_2} (\omega_\ell, \ell) = \frac{ \mathfrak W [ \uhr, u_{WKB_1} ] (\omega_\ell, \ell)}{\mathfrak W[u_{WKB_2},u_{WKB_1}] (\omega_\ell, \ell)}
\end{align}
which satisfy 
\begin{align}
&| A_{WKB_1} (\omega_\ell, \ell) |\sim \ell^{-1} |\mathfrak W [ \uhr, u_{WKB_2} ](\omega_\ell,\ell, R_1) |  \sim\ell^{-1} |\omega_\ell| \sim  1 \label{eq:estimonawkb1}\\
&| B_{WKB_2} (\omega_\ell, \ell) |\sim \ell^{-1} |\mathfrak W [ \uhr, u_{WKB_1} ] (\omega_\ell,\ell,R_1) |    \lesssim e^{-\tilde c_2 \ell}   \label{eq:estimonawkb2}
\end{align}
for some $\tilde c_2 >0$ in view of \eqref{eq:v1}, \eqref{eq:v2} and the WKB bounds of \eqref{eq:estimatesoneta}.

\paragraph{Region III: $R_2 \leq r^\ast < \infty$}
This region is completely analogous to Region I. Indeed,  analogous to \eqref{eq:estimonawkb1} and \eqref{eq:estimonawkb2} we have 
\begin{align}
& |\mathfrak W [ u_{WKB_1} , \uchr ] (\omega_\ell, \ell) | =  |\mathfrak W [u_{WKB_1} , \uchr ] (\omega_\ell,\ell, R_2) |  \sim \ell     ,\\
& |\mathfrak W[  u_{WKB_2}, \uchr ] (\omega_\ell, \ell)|= |\mathfrak W[ u_{WKB_2}, \uchr ](\omega_\ell,\ell, R_2) |\lesssim e^{-\tilde c_3\ell}   
\end{align}
for some $\tilde c_3 >0$. 

Now, using \eqref{eq:expansionofuhr} we expand
 \begin{align} \nonumber
 	\mathfrak W[\uhr, \uchr]  (\omega_\ell, \ell) &= A_{WKB_1}  (\omega_\ell, \ell) \mathfrak W [ u_{WKB_1} , \uchr ] (\omega_\ell, \ell) 
 	\\
 	&+B_{WKB_2}   (\omega_\ell, \ell)\mathfrak W[ u_{WKB_2}, \uchr ] (\omega_\ell, \ell).\label{eq:decomposition}
 \end{align}
From \eqref{eq:decomposition} and putting everything together we arrive at
\begin{align}
| 	\mathfrak W[\uhr, \uchr] | (\omega_\ell, \ell) \sim \ell 
\end{align} 
which shows \eqref{eq:lowerboundontT} as $\ell \sim |\omega_\ell|$.
\end{proof}
\end{lemma}
\section{Proof of the main result}

\label{sec:proofofmaintheorem}

Before we prove our main theorem, we show that the quasinormal modes of \cref{thm:existenceofqnm} extend smoothly to the interior and   arise from smooth initial data.

\begin{prop}\label{prop:qnmextendssmoothly}
	The quasinormal modes defined in \cref{thm:existenceofqnm} extend smoothly to $\mathcal B \cup \mathcal H_R \cup \mathcal R$  as solutions to \eqref{eq:waveequation}. In outgoing Eddington--Finkelstein coordinates $(v,r,\theta,\varphi)$ they have the form
	\begin{align}
		\psi_\ell (v,r,\theta,\varphi) = e^{-i \omega_\ell v } e^{i \omega_\ell f_\ast(r)} \tilde R_{\mathcal H_R}(\omega_\ell,\ell, r) Y_{\ell,0}(\theta),
	\end{align}
	where $\tilde R_{\mathcal H_R}(\omega_\ell,\ell, r)$ is as in \cref{lem:defnrs} and $f_\ast$ is as in \eqref{eq:firsttimefastarises}. Moreover, $\psi_\ell\vert_{D^+(\Sigma)}$ can be considered as arising from smooth data  $(\Psi_0,\Psi_1):=  (\psi_\ell \vert_\Sigma, \hat n_{\Sigma} \psi_\ell|_{\Sigma}) \in C^\infty(\Sigma)\times C^\infty(\Sigma)\cap \underline H_0^1(\Sigma) \times \underline L^2(\Sigma)$ posed on $\Sigma$.
	\begin{proof}
		Let $R_\ell(r)$ be the radial part of the quasinormal mode from  \cref{thm:existenceofqnm}. Then, $R_\ell$ is a solution to \eqref{eq:radialode1} and satisfies the boundary condition \ref{item1}  of  \cref{thm:existenceofqnm} at $r=r_+$. Thus, we  have that     $\tilde R_\ell(r) = \tilde R_{{\mathcal H}_R}(\omega_\ell,\ell,r) $ as well as  $R_\ell(r) =  R_{{\mathcal H}_R}(\omega_\ell,\ell,r)$ from which we obtain that 
		\begin{align}\label{eq:expressiontoextend}
			\psi_\ell (v,r ,\theta,\varphi) = e^{-i \omega_\ell v }e^{i \omega_\ell f_\ast(r)}  \tilde R_{{\mathcal H}_R}(\omega_\ell, \ell, r) Y_{\ell,0}(\theta) 
		\end{align}
		for $r>r_+$. We now note that we can extend  \eqref{eq:expressiontoextend} smoothly (in fact analytically)  to  $r_- < r\leq r_+$  as a solution to \eqref{eq:waveequation}. 
			Doing so, we obtain that $\psi_\ell$ is smooth (in fact analytic) on $\mathcal B \cup \mathcal H_R \cup \mathcal R$, such that in particular \begin{align}(\Psi_0,\Psi_1):=  (\psi_\ell \vert_\Sigma, \hat n_{\Sigma} \psi_\ell|_{\Sigma}) \in C^\infty(\Sigma)\times C^\infty(\Sigma).\end{align}
			
	We are left to show $(\Psi_0,\Psi_1) \in \underline H^1_0(\Sigma) \times \underline L^2(\Sigma)$ for which it suffices to consider the asymptotic region $\Sigma\cap \{ r\geq R\}= \{t=0\}\cap \{ r\geq R\}$. We first consider $\Psi_1$ and note that \begin{align}|\Psi_1 |^2 =| \hat n_{\Sigma} \psi_\ell |^2= |r  \sqrt{\frac{r^2}{\Delta}} \partial_t   \psi_\ell |^2 \sim | (- i \omega_\ell) \psi_\ell |^2 \sim |\omega_\ell|^2 |\psi_\ell|^2  \end{align} for $r\geq R$. 
			Hence,  \begin{align}
			|\Psi_1 |^2 r^{-1} \dvol_{\Sigma} \sim  |\omega_\ell|^2 |\psi_\ell|^2 \d r \d \sigma_{\mathbb S^2} \sim  |\omega_\ell|^2 r^{-3 - 2 \sqrt{\frac 94 - \alpha} } |\tilde R_{D}(\omega_\ell ,\ell,r)|^2 |Y_{\ell, 0 }|^2 \d r \d \sigma_{\mathbb S^2}
		\end{align}
		for $r\geq R$. Here, we have used that $\lim_{r\to\infty} r^{\frac 32 - \sqrt{\frac 94 - \alpha} }R_\ell (  r) = 0$ from   \cref{thm:existenceofqnm} such that \cref{lem:defnrs} implies that 
		\begin{align}R_\ell ( r) = \lambda R_{D}(\omega_\ell,\ell,r)\end{align} for some constant $\lambda \neq 0$.
		Thus, $\| \Psi_1 \|_{\underline L^2(\Sigma) } < \infty$ as $\lim_{r\to\infty} \tilde R_D (\omega_\ell,\ell,r) =1$.  
		
		Toward  $ \|\Psi_0 \|_{\underline H_0^1(\Sigma)} < \infty$ we first consider the radial term of $|\tilde \nabla \psi_\ell|^2 r \dvol_{\Sigma}$  for $r\geq R$:
		\begin{align} \nonumber 
			g^{rr} f^2 | \partial_r (f^{-1} \psi_\ell) |^2  r \dvol_{\Sigma} & \sim r^2 r^{-3+ 2 \sqrt{\frac 94 - \alpha}} |\partial_r (r^{   -2 \sqrt{\frac 94 - \alpha}}  \tilde R_{D})    |^2 |Y_{\ell, 0}|^2 r^2 \d r \d \sigma_{\mathbb S^2}\\
			& \lesssim (|\tilde R_D|^2 +  r^2 | \partial_r \tilde R_D|^2) r^{-1 - 2\sqrt{\frac 94 - \alpha}} |Y_{\ell, 0}|^2 \d r \d \sigma_{\mathbb S^2}.\label{eq:estimateongrr}
		\end{align}
		Now, from \cref{lem:defnrs} note that $\tilde R_D$ and $r^2 \partial_r \tilde R_D =  r^2 \frac{\d s}{\d r} \partial_s \tilde R_D= - \partial_s \tilde R_D$ for $s=\frac 1r$ remain bounded at $r=\infty$ from which we obtain that \eqref{eq:estimateongrr} is integrable. Similarly, we also obtain that the angular derivatives are integrable and we conclude that $ \|\Psi_0 \|_{\underline H_0^1(\Sigma)} < \infty$.
	\end{proof}
\end{prop}

Now, we finally proceed to the proof of \cref{thm:maintheorem}.

\begin{proof}[Proof of {\cref{thm:maintheorem}}]
From \cref{prop:qnmextendssmoothly} we have that for each $\ell$ sufficiently large, the  quasinormal mode $\psi_\ell(v,r,\theta,\varphi)$ is a smooth solution to \eqref{eq:waveequation} for $r\in (r_-,\infty)$, $v \in \mathbb R$, $(\theta,\varphi)\in \mathbb S^2$ and satisfies   
\begin{align}
		\psi_\ell (v,r ,\theta,\varphi) = e^{-i \omega_\ell v }e^{i \omega f_\ast(r)}  \tilde R_{{\mathcal H}_R}(\omega_\ell, \ell, r) Y_{\ell,0}(\theta) 
\end{align}
for some $\tilde R_{\mathcal H_R}(\omega_\ell, \ell, r)$ which is analytic for $r\in (r_-,\infty)$. 
We   also showed in \cref{prop:qnmextendssmoothly} that each $\psi_\ell$ arises from initial data $(\Psi_0,\Psi_1) \in C^\infty(\Sigma)\times C^\infty(\Sigma)\cap \underline H_0^1(\Sigma) \times \underline L^2(\Sigma)$.

Now, in view of \eqref{eq:defnuhra}--\eqref{eq:defnuchla} and \cref{lem:defnofRandT}, we have 
\begin{align}\label{eq:termstolookat}
	\psi_\ell(t,r,\theta,\varphi)   = \frac{r_+ e^{i \omega_\ell f_\ast(r_+) }}{r_-  } e^{-i\omega_\ell t} \left(  \frac{\mathfrak T (\omega_\ell,\ell)}{ e^{i\omega_\ell g_\ast(r_-) } }R_{\mathcal{CH}_L} (\omega_\ell ,\ell,r)+ \frac{ \mathfrak R(\omega_\ell ,\ell) }{ e^{-i\omega_\ell g_\ast(r_-) } } R_{\mathcal{CH}_R}(\omega_\ell ,\ell,r) \right) Y_{\ell,0}(\theta) .
\end{align}
We first note that  $e^{-i \omega_\ell t }  R_{\mathcal{CH}_R}(\omega_\ell, \ell, r)Y_{\ell,0}(\theta)$ appearing in the   second term extends smoothly across $\mathcal{CH}_R$. Indeed, 
\begin{align} e^{-i\omega_\ell t }  R_{\mathcal{CH}_R} (\omega_\ell, \ell, r) = e^{-i\omega_\ell t } (r - r_-)^{ \frac{ i \omega_\ell}{2\kappa_-}  } \tilde R_{\mathcal{CH}_R} (\omega_\ell ,\ell,r) = e^{i \omega_\ell u} e^{ - i  \omega_\ell g_\ast(r)} \tilde R_{\Ch_R}(\omega_\ell ,\ell,r),\end{align}
 where    $\tilde R_{\Ch_R}(\omega_\ell, \ell, r)$ and $e^{-i\omega_\ell g_\ast(r)}$ extend smoothly to $r=r_-$ in view of \cref{lem:defnrs}. 

Thus, in order to address the regularity of $\psi_\ell$ at the Cauchy horizon, we now focus on the first term of \eqref{eq:termstolookat}. We compute 
\begin{align} \nonumber 
	e^{-i \omega_\ell t } R_{\mathcal{CH}_L} (\omega_\ell ,\ell,r) &
	=  	e^{-i \omega_\ell t } (r-r_-)^{ -\frac{i\omega_\ell}{2 \kappa_- }} \tilde R_{\mathcal{CH}_L} (\omega_\ell ,\ell,r).\\ \nonumber
	& =  (r-r_-)^{ \frac{-i\omega_\ell}{ \kappa_- }}	e^{-i \omega_\ell (t-r^\ast) }   e^{ - i  \omega_\ell g_\ast(r)} \tilde R_{\mathcal{CH}_L} (\omega_\ell ,\ell,r)\\
	& = (r-r_-)^{ \frac{-i\omega_\ell}{ \kappa_- }}	e^{ i \omega_\ell u } e^{ - i  \omega_\ell g_\ast(r)} \tilde R_{\mathcal{CH}_L} (\omega_\ell ,\ell,r).
\end{align}
As above, $\tilde R_{\mathcal{CH}_L} (\omega_\ell ,\ell,r)$, $e^{ - i  \omega_\ell g_\ast(r)}$ and  $		e^{ i \omega_\ell u }$ extend smoothly to $\Ch_R$. Moreover, in view of \eqref{eq:h1dotnorm}  and since $\mathfrak T (\omega_\ell,\ell)\neq 0$ in view of \cref{eq:transmissioncoefnotzero},  it suffices to show that 
\begin{align}\label{eq:blowupofr-r-}
	\| (r-r_-)^{ \frac{-i\omega_\ell}{ \kappa_- }}	\|_{H^1([r_-, r_-+\epsilon])} =  \infty 
\end{align}
for some $\epsilon>0$. A direct computation shows that  \eqref{eq:blowupofr-r-} holds true for \begin{align} \label{eq:boundbeta} \beta:= \frac{-\operatorname{Im}(\omega_\ell) }{ \kappa_- } <   \frac{1}{2}\end{align}  which holds true for any  $\ell\geq \tilde \ell $ for some $ \tilde \ell$ sufficiently large.

Finally we show the last claim of \cref{thm:maintheorem} regarding the genericity. We have shown that for each $\ell\geq \tilde \ell$ for some $\tilde \ell$ sufficiently large, we obtain a solution which fails to be in $H^1_{\mathrm{loc}}$ at the Cauchy horizon. Further, by   orthogonality  of the spherical harmonics in $L^2(\mathbb S^2)$, we also have that    any non-trivial linear combination of elements in $\{ \psi_\ell\colon \ell\geq \tilde\ell\} $   satisfies \eqref{eq:blowup}. Thus, there exists an infinite dimensional subspace of the space of initial data $C^\infty(\Sigma)\times C^\infty(\Sigma)\cap \underline H_0^1(\Sigma) \times \underline L^2(\Sigma)$  with the property that any non-zero element leads to a solution satisfying \eqref{eq:blowup}.  
   This shows that the space of initial data $H \subset C^\infty(\Sigma)\times C^\infty(\Sigma)\cap \underline H_0^1(\Sigma) \times \underline L^2(\Sigma)$ for which \eqref{eq:blowup} is false has infinite codimension.
   This concludes the proof. 
\end{proof}
\printbibliography[heading=bibintoc]
\end{document}